\documentclass[11pt,psfig]{article}
\usepackage{graphics}
\usepackage[ruled,vlined]{algorithm2e}

\usepackage{enumerate}
\usepackage{graphicx}
\usepackage{amssymb}
\usepackage{amsmath}
\usepackage{amsthm}
\usepackage{amsfonts}
\usepackage{mathtools}
\usepackage{xcolor}
\usepackage{amsfonts}
\usepackage{amssymb}
\usepackage{bm}
\usepackage[numbers]{natbib}
\usepackage{hyperref}
\usepackage{academicons}
\usepackage{xcolor}
\usepackage{color}
\usepackage{tikz}
\numberwithin{equation}{section}
\newtheorem{theorem}{Theorem}[section]

\newcommand{\ket}[1]{|#1\rangle}
\newcommand{\bra}[1]{\langle#1|}
\newcommand{\Tr}{\text{Tr}}

\newcommand{\x}{\bm{x}}
\newcommand{\dd}{\bm{d}}
\newcommand{\e}{\bm{e}}
\newcommand{\s}{\bm{s}}
\newcommand{\y}{\bm{y}}
\newcommand{\ii}{\text{i}}

\newcommand{\R}{\mathbb R}

\usepackage{sidecap,tikz}
\DeclareRobustCommand{\orcidicon}{\hspace{-1.0mm}
	\begin{tikzpicture}
		\draw[lime, fill=lime] (0.0,0.0) 
		circle [radius=0.15] 
		node[white] {{\fontfamily{qag}\selectfont \tiny \,ID}};
		\draw[white, fill=white] (-0.0525,0.095) 
		circle [radius=0.007];
	\end{tikzpicture}
	\hspace{-3.0mm}
}
\newcommand{\orcidDM}{\href{https://orcid.org/0000-0002-0166-4760}{\orcidicon}}
\newcommand{\orcidVAM}{\href{https://orcid.org/0000-0002-9015-1234}{\orcidicon}}

\usepackage{mathtools}
\DeclarePairedDelimiterX\braket[2]{\langle}{\rangle}{#1\,\delimsize\vert\,\mathopen{}#2}

\graphicspath{{figs/}} 

\begin{document}
\date{\small\textsl{\today}}
\title{Numerical solution of quantum Landau–Lifshitz–Gilbert equation}
\author{
{\large Vahid Azimi-Mousolou\orcidVAM}\thanks{Email address: vahid.azimi-mousolou@physics.uu.se} \\ 
            Department of Physics and Astronomy, Divition of Material Theory, \\ Uppsala University, Box 516, 
            SE-751 20 Uppsala, Sweden \\ \\
{\large Davoud Mirzaei\orcidDM}\thanks{Emial address: davoud.mirzaei@it.uu.se} \\
            Department of Information Technology, Division of Scientific Computing, \\ Uppsala University, 
            Box 337, SE-751 05 Uppsala, Sweden
 }

\begingroup
\renewcommand\thefootnote{\#}
\footnotetext{The authors contributed equally to this work.}
\endgroup
\maketitle
\begin{abstract}
The classical Landau–Lifshitz–Gilbert (LLG) equation has long served as a cornerstone for modeling magnetization dynamics in magnetic systems, yet its classical nature limits its applicability to inherently quantum phenomena such as entanglement and nonlocal correlations. Inspired by the need to incorporate quantum effects into spin dynamics, recently a quantum generalization of the LLG equation is proposed [Phys. Rev. Lett. 133, 266704 (2024)] which captures essential quantum behavior in many-body systems. In this work, we develop a robust numerical methodology tailored to this quantum LLG framework that not only handles the complexity of quantum many-body systems but also preserves the intrinsic mathematical structures and physical properties dictated by the equation. We apply the proposed method to a class of many-body quantum spin systems, which host topological states of matter, and demonstrate rich quantum behavior, including the emergence of long-time entangled states. This approach opens a pathway toward reliable simulations of quantum magnetism beyond classical approximations, potentially leading to new discoveries. 
\end{abstract}

\section{Introduction}\label{sect:intro}

The classical Landau-Lifshitz (LL) equation \cite{landau35} and its extension, the Landau-Lifshitz-Gilbert (LLG) equation \cite{gilbert04}, are fundamental tools in atomistic spin dynamics and provide a powerful framework for investigating the microscopic behavior of magnetic systems and materials. These equations, which are mathematically equivalent in the classical regime \cite{Lakshmanan1984, liu24}, have been pivotal in advancing our understanding of diverse magnetic phenomena. Notable applications include the magnetization dynamics of topological structures \cite{pereiro14}, ultrafast demagnetization in face-centered cubic (fcc) 
Ni during pump-probe experiments \cite{evans15, pankratova22}, magnetization reversal in ferrimagnetic FeGd alloys \cite{chimata15}, and the quantitative characterization of spin texture complexity \cite{bagrov20}.

Despite their extensive contributions, classical approaches such as the LL and LLG equations exhibit inherent limitations. Most notably, the Bohr–van Leeuwen theorem, established over a century ago, demonstrates that classical statistical mechanics cannot account for magnetic ordering and shows that magnetism is fundamentally a quantum mechanical phenomenon. A common justification for employing classical equations has been the assumption that quantum descriptions converge to classical behavior in the limit of large angular momentum. However, in real materials, atomic angular momenta are finite, typically ranging from 0.5 to 2.5 in transition metals, and up to 3.5 in certain rare-earth elements. More critically, essential quantum features such as many-body quantum correlations and entanglement, which are central to the physics of quantum systems, including magnetic domains \cite{Laurell2021, Mathew2020, Menon2023, mousolou20, Scheie2021, Scheie2024}, are entirely beyond the scope of classical equations. To achieve a deeper and more accurate understanding of magnetic dynamics, it is therefore essential to investigate quantum mechanical analogs of the LL and LLG equations.

Wieser \cite{wieser13} proposed a quantum analog of the Landau–Lifshitz (LL) equation based on a phenomenological open-system approach, to derive the classical LL equation and its dynamics from a more fundamental quantum perspective. Expanding on this idea, in a recent study \cite{liu24}, a quantum analog of the classical LLG equation has been demonstrated. The new model provides several nontrivial quantum effects in systems consisting of spin-$\frac{1}{2}$ particle pairs such as spinless local states emerging in antiferromagnetically coupled particles and the formation of highly entangled states in the long-time limit, despite the presence of dissipation. These findings lay the groundwork for a new paradigm in exploring the quantum mechanical aspects of spin dynamics in magnetic materials. All results presented in \cite{liu24} are restricted to single- and two-spin systems, whereas quantum effects tend to be more pronounced and complex in multi-spin systems. Although exact solutions are available in certain special cases, a robust numerical solver capable of addressing the model in its general form, particularly for multi-spin systems, is essential for advancing both theoretical understanding and practical applications.

When it comes to numerical simulation, some computational tools and software packages exist for solving the classical LL and LLG  equations, such as the UppASD code \cite{eriksson17, skubic08}, which can handle large-scale spin systems. These classical methods support parameter calculations for general spin Hamiltonians and incorporate interactions such as Heisenberg exchange, symmetric and antisymmetric anisotropic exchange, and magnetic anisotropy constants. However, no established computational framework currently exists for solving the quantum analogs of the LL and LLG equations, which are essential for exploring quantum spin dynamics. A major challenge arises from the fact that quantum equations scale exponentially in size and complexity compared to their classical counterparts and require numerical solvers that preserve essential structural properties of the solution.

In this work, we establish a numerical methodology for solving quantum LLG equations in multi-body spin systems and apply this approach to investigate entanglement dynamics and emergent quantum behavior in quantum magnetic materials. The method is specifically designed to overcome the complexity of quantum multi-body dynamics while preserving key physical and mathematical structures of the quantum LLG equation. We demonstrate the effectiveness of the method by simulating quantum spin systems that host topological spin textures, including skyrmions, and reveal the formation of long-lived entangled states. These results provide a concrete step toward 
scalable, structure-preserving simulations of quantum spin dynamics beyond classical approximations.

\section{Quantum Landau-Lifshitz-Gilbert (q-LLG) equation}

While the classical LL and LLG equations describe magnetization as a classical vector, the underlying physical degrees of freedom originate from quantum spins density operators \cite{liu24, wieser13}. 
The quantum version of the LLG equation (q-LLG) has been recently demonstrated by \cite{liu24} 
\begin{equation}\label{eq:q-LLG}
\dot{\rho} = \frac{\text{i}}{\hslash}[\rho, H] + \text{i}{\kappa} [\rho,\dot{\rho}]
\end{equation}
where $\rho\in \mathbb C^{N\times N}$ (a hermitian, positive semi-definite complex matrix of size $N\times N$ with unit trace) is the density operator, 
$\dot{\rho} = \frac{d \rho }{dt}$, 
$\hslash$ is the Planck constant, and the commutator, $[A,B]$, of two operators $A$ and $B$, is defined to be $[A,B]=AB-BA$.
The second term on the right-hand side of \eqref{eq:q-LLG} has a dissipative (damping-like) character, with $\kappa$ representing the dimensionless damping rate.

The q-LLG equation governs the dynamics of the density operator $\rho$, which describes the quantum state of given spin system. The dynamics of a simple single spin-$\frac{1}{2}$ system and a dimer system composed of two spin-$\frac{1}{2}$ particles have already been discussed in \cite{liu24}. In the special case of a single spin, the q-LLG \cite{wieser13} and q-LL \cite{liu24} equations are equivalent up to a time rescaling. This mirrors the relationship between their classical counterparts. However, in multi-spin systems, the q-LLG equation exhibits significant differences from  the q-LL equation \cite{liu24}. 

The exact solution to equation \eqref{eq:q-LLG} presented in \cite{liu24} is limited to simple cases with specific initial conditions. We also note that solving this equation becomes significantly more challenging in many-body systems due to the exponential growth in dimensionality. For instance, in the lowest-dimensional spin-$\frac{1}{2}$ (qubit) system, where $N = 2^{n}$ with $n$ denoting the number of spin-$\frac{1}{2}$ particles, the density matrix resides in a space of dimension $4^{n}-1$. This exponential scaling, combined with the inherent nonlinearity of the equation, poses substantial computational challenges for simulating large quantum systems.

This paper aims to develop a robust numerical procedure for solving equation \eqref{eq:q-LLG}, as no such method currently exists in the literature.
The solution of equation \eqref{eq:q-LLG} satisfies certain properties that are particularly valuable for developing a reliable numerical method. One key property is that the density operator $\rho(t)$ is {\em Hermitian} and {\em positive semi-definite} for all $t\geq 0$. More generally,  
this variable possesses 
the conservation of the spectrum, i.e., 
\begin{equation}\label{eq:spectrum_conservation}
\frac{d}{dt} \lambda_k(\rho(t)) = 0
\end{equation}
where $\lambda_k\equiv \lambda_k(\rho(t))$, $k=1,\ldots,N$ are eigenvalues of $\rho(t)$. This means that the spectrum  is independent of time; $\lambda_k(\rho(t)) = \lambda_k(\rho(0))$ for all $t\geq 0$. 
As a consequence, the traces of all powers of $\rho$ are conserved;   
\begin{equation*}\label{eq:trace_conservation}
\frac{d}{dt} \text{Tr}(\rho^m) = 0, \quad m=1,2,\ldots  . 
\end{equation*}
Note that, cases $m = 1$ and $m = 2$ are sometimes refered to as conservation of trace and conservation of purity, respectively. See the supplementary draft of \cite{liu24} for proofs. 

To guarantee the accuracy and physical consistency of the numerical scheme, it is essential that the proposed method preserves these conservation properties.
Below we develop a numerical ODE solver for \eqref{eq:q-LLG} which makes it possible to solve the equation for an arbitrary given Hamiltonian $H$ and initial state $\rho_0 = \rho(t=0)$, while preserving the structural and physical properties of the q-LLG equation. 

\section{Converting to a standard form}

In this section, we describe how the q-LLG equation $ \eqref{eq:q-LLG} $ can be reformulated in an equivalent form to be solvable using standard ODE solvers. First, we rewrite equation \eqref{eq:q-LLG}  as   
\begin{equation}\label{eq:q-LLG-2}  
   (I - \ii \kappa \rho) \dot \rho  + \dot \rho(\ii \kappa \rho) = \frac{\ii}{\hslash} [\rho,H].  
\end{equation}  
On the other hand, the spectral decomposition (eigendecomposition) of Hermitian matrix $ \rho $ is given by  
$$  
\rho = V\Lambda V^*  
$$  
where  
$$  
\Lambda  = \begin{bmatrix}  
\lambda_1 &          &       &0\\  
          &\lambda_2 &       &\\  
          &          & \ddots &\\  
  0        &          &        & \lambda_N  
\end{bmatrix}\in \mathbb R^{N\times N}, \quad V = \begin{bmatrix}  
\bm{v}_1 \; \bm{v}_2  \; \cdots  \; \bm{v}_N  
\end{bmatrix} \in \mathbb C^{N\times N}  
$$  
are the eigenvalues and eigenvectors of $ \rho $. Substituting this decomposition into $ \eqref{eq:q-LLG-2} $, we obtain  
\begin{equation*}  
    (I+S) V^* \dot \rho\, V - V^*\dot \rho\, V S = \frac{\ii}{\hslash}  V^* [\rho,H] V  
\end{equation*}  
for $ S = - \ii \kappa \Lambda $, i.e.,  
$ s_\ell = -\ii\kappa \lambda_\ell $ for $ \ell=1,2,\ldots,N$.  
Next, we define  
\begin{equation}\label{eq:Xdef}
    X := V^*\dot \rho\, V  
\end{equation}    
which leads to  
\begin{equation}\label{eq:q-LLG-3}  
    (I+S) X - X S = D  
\end{equation}  
where the right-hand side is given by $ D = \frac{\ii}{\hslash}  V^* [\rho,H] V $.  
Equation $ \eqref{eq:q-LLG-3} $ is a {\em Sylvester equation} with diagonal matrices $ S $ and $ I-S $. 
Since $ \rho $ is unknown, all matrices in this system are also unknown. However, let us explore how we can solve this equation for $X$ using a linear algebra technique.

Assuming $ X $ is expressed in terms of its columns as  
$ X = [\x_1 \; \x_2 \; \cdots \; \x_N] $  
and similarly, $ D $ in terms of its columns as  
$ D = [\dd_1 \; \dd_2 \; \cdots \; \dd_N] $,  
we can rewrite equation $ \eqref{eq:q-LLG-3} $ as  
\begin{equation*}  
    \big[(I+S) \x_1 \; (I+S)\x_2 \; \cdots \; (I+S)\x_N\big] - \big[s_1 \x_1 \; s_2\x_2 \; \cdots \; s_N \x_N\big] = \big[\dd_1 \; \dd_2 \; \cdots \; \dd_N\big]  
\end{equation*}  
which corresponds to a set of $ N $ decoupled diagonal systems  
\begin{equation*}  
    (I+S - s_\ell I) \x_\ell = \dd_\ell, \quad \ell = 1,2,\ldots,N  
\end{equation*}  
with solutions  
\begin{equation}\label{eq:x_ell_sol}  
    \x_\ell = \text{cwd}(\dd_\ell, \e + \s - s_\ell \e)  , \quad \ell = 1,2,\ldots,N 
\end{equation}  
where $ \s $ is the diagonal of $ S $, and $ \e = [1,1,\ldots,1]^T $.  
Here, $ \text{cwd}(\x,\y) $ denotes the component-wise division of vectors $ \x $ and $ \y $. The solution is well-defined because the denominators $ \e + \s - s_\ell \e $ are nonzero unless $ s_\ell - s_j = 1 $ for some $ \ell $ and $ j $. However, this never occurs because $ s_\ell-s_j = -\ii\kappa (\lambda_\ell-\lambda_j)$, and $\kappa, \lambda_j,\lambda_\ell\in \mathbb R$. 

We remark that equation \eqref{eq:x_ell_sol} has benefited from the ODE \eqref{eq:q-LLG} to express $X$ in terms of $\rho$ only (via $V$ and $\Lambda$). Now, we can write from \eqref{eq:Xdef}, 
\begin{equation}\label{eq:q-LLG-final}  
    \begin{split}  
    &\dot \rho = VXV^*  \\  
    & \rho(0) = \rho_0  
    \end{split}  
\end{equation}  
where $X$ is obtained from \eqref{eq:x_ell_sol}. 
Note that the right-hand side $ VXV^* $  is now independent of $ \dot \rho $ and depends only on $ \rho $ via the eigendecomposition and equation \eqref{eq:x_ell_sol}. 

The reformulated system in \eqref{eq:q-LLG-final} defines an initial value problem suitable for standard ODE solvers. However, most built-in solvers are designed to handle vector-valued systems, whereas in this case, the solution $\rho$ is matrix-valued. To apply such solvers, one common approach is to flatten $\rho$ into a long vector at each time step, solve the corresponding vector-valued ODE, and then reshape the result back into matrix form. 

In this work, however, we do not use built-in solvers. Instead, we employ a class of Runge-Kutta methods directly in the matrix formulation and enhance them to ensure that the structural properties of the solution such as symmetry, non-negativity, and traces are preserved throughout the integration process.

\section{Explicit Runge-Kutta methods}

Let's get started with the Euler's method, the simplest explicit ODE solver which begins with the initial density matrix $\rho_0$ and computes an approximate solution at each time level $t_{k+1} = (k+1)h$ using the recurrence
\begin{equation*}\label{euler_method}
\rho_{k+1} = \rho_k + h \, V_k X_k V_k^*, \quad k = 0, 1, 2, \ldots,
\end{equation*}
where $h=t_{k+1}-t_k$ denotes the time step, $V_k$ is the eigenmatrix of $\rho_k$, and $X_k$ is obtained from equation \eqref{eq:x_ell_sol}, with $S$ and $D$ derived from $\rho_k$. The iteration continues until the final time $t_F$ is reached. 
The Euler method  can be written as
$$
\rho_{k+1} = \rho_k + h f(\rho_k), \quad k = 0, 1, 2, \ldots,
$$
where $f(\rho_k) = V_k X_k V_k^*$. This method gives a first-order approximation provided that the solution $\rho$ is sufficiently smooth.
This basic scheme can be extended to a more general explicit method
$$
\rho_{k+1} = \rho_k + h \psi(\rho_k; h), \quad k = 0, 1, 2, \ldots,
$$
where $\psi(\rho; h)$ is a function that replaces $f(\rho)$ and is assumed to be (at least) continuous in $h$ and Lipschitz continuous in $\rho$. Imposing additional regularity and consistency conditions on $\psi$ enables the construction of higher-order schemes.

A prominent family of such high-order methods is the class of explicit RK methods. An explicit $s$-stage RK method has the form
\begin{equation}\label{rk:sstageform}
\begin{array}{rl}
\displaystyle z_1 &=\, \rho_k\\
\displaystyle z_2 &=\, \rho_k + h a_{2,1} f(z_1) \\
\displaystyle z_3 &=\, \rho_k + h \left[ a_{3,1} f(z_1) + a_{3,2} f(z_2) \right] \\
&\vdots \\
\displaystyle z_s &=\, \rho_k + h \left[ a_{s,1} f(z_1) + \cdots + a_{s,s-1} f(z_{s-1}) \right] \\  \hline 
\displaystyle \rho_{k+1} &=\, \rho_k + h \left[ b_1 f(z_1) + b_2 f(z_2) + \cdots + b_s f(z_s) \right]
\end{array}
\end{equation}
The method is fully defined by the real coefficients $\{a_{\ell,j}, b_j\}$, which are usually presented in a compact form called the {\em Butcher tableau} (see, e.g., \cite{Butcher:2016-1}). At minimum, these coefficients must satisfy
$$
\sum_{j=1}^s b_j = 1.
$$
The Euler method is simply a one-stage RK method and provides first-order accuracy, i.e., the error decays as $\mathcal{O}(h^1)$. Among higher-order methods, the most widely used is the classical 4-stage RK method, which gives fourth-order accuracy, i.e., $\mathcal{O}(h^4)$.

However, all explicit RK methods are conditionally stable, meaning that the time step $h$ must be sufficiently small to prevent numerical instability and blow-up. Although implicit methods allow for larger time steps and improved stability properties, their application to the matrix-valued evolution equations considered here would be computationally expensive due to the complex structure of the right-hand side. Therefore, we restrict ourselves to explicit methods with suitably small step sizes.

In what follows, we establish a couple of structural properties of the solution that are preserved by the explicit RK methods applied to this class of equations.

\begin{theorem}
    The explicit RK schemes  \eqref{rk:sstageform} preserve the Hermiticity of the density matrices $\rho_k$ for $k\geq 0$.
\end{theorem}

\begin{proof}
We show that if $\rho_k$ is Hermitian, then the next iterate $\rho_{k+1}$ computed by \eqref{rk:sstageform} remains Hermitian. We prove that $f(z_j)$ is Hermitian for all $j=1,\ldots,s$. It suffices to show that $X$ matrices at each stage are Hermitian. For $j=1$
recall that $X$ is computed as the solution of equation \eqref{eq:q-LLG-3}. Since both $\rho_k$ and $H$ are Hermitian matrices, their commutator $[\rho_k, H]$ is skew-Hermitian. Therefore,
$$
D_k = \frac{\ii}{\hslash} V_k^*[\rho_k, H]V_k
$$
is Hermitian. Given that $D_k$ is Hermitian, the equation \eqref{eq:q-LLG-3} has a unique Hermitian solution $X_k$. This shows that $z_2$ is Hermitian. Similarly, for $j=\ell>1$ we can show that $z_j$ and $f(z_j)$ are all Hermitian.  

\end{proof}

\begin{theorem}
    The explicit RK schemes  \eqref{rk:sstageform} preserve the trace of $\rho_k$ for $k\geq 0$. 
\end{theorem}
\begin{proof}
    According to \eqref{rk:sstageform} it is enough to show that 
    $\mathrm{Tr}(f(z_j)) = 0$ for all $j=1,\ldots,s$. 
    Starting with $j=1$, we show that $\mathrm{Tr}( f(z_1)) =  \mathrm{Tr}( V_k X_k V_k^*) = 0$. Using the known property $\mathrm{Tr}(AB) = \mathrm{Tr}(BA)$, we have
    $
    \mathrm{Tr}( V_k X_k V_k^*) = \mathrm{Tr}(X_k)
    $. We observe from \eqref{eq:q-LLG-3} that
    $\mathrm{Tr}(X_k) = \mathrm{Tr}(X_kS)-\mathrm{Tr}(SX_k)+ \mathrm{Tr}(D_k) = 0$ because the trace of commutator $[\rho_k,H]$ is zero. This proves that $\mathrm{Tr}(z_2)=\mathrm{Tr}(\rho_k)$. Similarly, we can show that $\mathrm{Tr}(f(z_j)) = 0$ and $\mathrm{Tr}(z_j) = \mathrm{Tr}(\rho_k)$ for $j>1$. 
\end{proof}

However, it can be shown that explicit RK methods do not preserve the trace of higher powers of $\rho_k$, i.e., $\mathrm{Tr}(\rho_k^n)$ for $n > 1$. This indicates a loss of exact conservation under RK time integration, even though $\rho(t)$, the exact solution of equation \eqref{eq:q-LLG} (or equivalently \eqref{eq:q-LLG-final}), satisfies such conserved quantities.

Another important property of the exact solution $\rho(t)$ is that its spectrum remains constant over time. In contrast, this spectral invariance is not preserved by explicit RK methods. As a result, it becomes non-trivial to prove that the numerical solution $\rho_k$ remains positive semi-definite for all $k \geq 1$.
In fact, the right-hand side terms $f(z_j)$ in the RK scheme are generally not positive semi-definite. Nevertheless, their contribution to the solution is scaled by the time step $h$, which helps to limit their impact on the eigenvalues of $\rho_k$ over short time intervals. Still, to achieve long-term stability and preserve key structural properties of the solution, it is necessary to develop a conservative numerical method.

\section{Structure-Preserving Solvers}

As previously mentioned, standard Runge-Kutta (RK) methods \eqref{rk:sstageform} fail to preserve important structural properties of the density matrix $\rho(t)$, such as non-negativity and the invariance of the trace of its higher powers. Theoretically, the spectrum of $\rho(t)$ remains constant over time, i.e. 
$$
\frac{d}{dt} \lambda_j(t) = 0, \quad j = 1,2,\ldots,N,
$$
where $\lambda_j(t) \in \mathbb{R}$ are the eigenvalues of $\rho(t)$. This spectral invariance implies other properties mentioned above. 

To guarantee that our numerical scheme respects this key spectral property, we propose a modified version of the standard explicit RK methods. The idea is to project the intermediate states of the RK procedure back onto the manifold of matrices with the same spectrum as the initial density matrix $\rho_0$.
Assume that $\rho_0$ admits the spectral decomposition
$$
\rho_0 = V_0 \Lambda_0 V_0^*,
$$
where $\Lambda_0 = \mathrm{diag}(\lambda_1, \ldots, \lambda_N)$ is a diagonal matrix of eigenvalues and $V_0$ is a unitary eigenmatrix. We now propose the following conservative RK scheme which preserves the spectrum of $\rho_0$ over the integration process:

\begin{equation}\label{rk:sstageform_modified}
\begin{array}{rl}
\displaystyle z_1 &=\, \rho_k\\
\displaystyle \tilde z_2 &=\, \rho_k + h a_{2,1} f(z_1) \\
\displaystyle z_2 & = \tilde V_2\Lambda_0 \tilde V_2^* \\
\displaystyle \tilde z_3 &=\, \rho_k + h \left[ a_{3,1} f(z_1) + a_{3,2} f(z_2) \right] \\
\displaystyle z_3 & = \tilde V_3\Lambda_0 \tilde V_3^* \\
&\vdots \\
\displaystyle \tilde z_s &=\, \rho_k + h \left[ a_{s,1} f(z_1) + \cdots + a_{s,s-1} f(z_{s-1}) \right] \\ 
\displaystyle z_s & = \tilde V_s\Lambda_0 \tilde V_s^* \\ 
\hline 
\displaystyle \tilde \rho_{k+1} &=\, \rho_k + h \left[ b_1 f(z_1) + b_2 f(z_2) + \cdots + b_s f(z_s) \right] \\
\displaystyle \rho_{k+1} & = \tilde V\Lambda_0 \tilde V^* \\
\end{array}
\end{equation}
where $\tilde V_j$ are eigenvector matrices of the intermediate solutions $\tilde z_j$ for $j = 2, \ldots, s$, and $\tilde V$ is the eigenvector matrix of $\tilde \rho_{k+1}$.

 The eigenvector matrices $\tilde V_j$ and $\tilde V$ are readily available from the computation of the right-hand side $f(\rho)$ at each stage, so their construction does not impose additional cost compared to the standard RK methods. In practice, we guarantee numerical consistency by sorting the eigenvectors according to their corresponding eigenvalues in increasing order. Furthermore, in our implementation (in MATLAB), we employ and adjust the more robust \texttt{svd} function to perform the eigendecomposition of Hermitian matrices instead of the default \texttt{eig} function.

\begin{theorem}
The modified Runge-Kutta methods defined in \eqref{rk:sstageform_modified} preserve the spectrum of the solution $\rho_k$ at each time step. Consequently, it also preserves non-negativity and the trace of all matrix powers of $\rho_k$.
\end{theorem}

This projection-based modification enforces structure preservation in each iteration, making it suitable for long-time simulations of quantum density matrices where physical constraints must be respected.

\section{Numerical results}
To assess the accuracy and reliability of the developed numerical method for solving the q-LLG equation \eqref{eq:q-LLG}, we apply it to a set of physical examples relevant to quantum spin dynamics. We focus on the dynamics of quantum entanglement in such spin systems through the q-LLG equation, as entanglement is a fundamental non-classical feature of quantum systems that cannot be addressed by classical approaches, including the classical LLG equation.

We consider a many-body quantum model system described by the following spin Hamiltonian 
\begin{equation}\label{eq:H1}
H = \frac{2J}{\hslash^2}\sum_{ij} \mathbf{S}_i \cdot \mathbf{S}_j + \frac{2}{\hslash^2}\sum_{ij} D_{ij} \cdot \left[\mathbf{S}_i \times \mathbf{S}_j \right] -\mu \sum_i B\cdot \mathbf{S}_i
\end{equation}
where $J\in\R$ is the isotropic Heisenberg exchange interaction, ${D}_{ij}\in \R^3$ ($D_{ij}=-D_{ji}$) are the Dzyaloshinskii–Moriya interaction (DMI), ${B}\in\R^3$ denotes an external magnetic field uniformly applied to all spins $\mathbf{S}_i$. The constant $\mu$ associated with the gyromagnetic ratio is given by $\mu = -\frac{\mu_B g}{\hslash}$, where $\mu_B$ is the Bohr magneton and $g$ is the Landé $g$-factor. The spin operators $\mathbf{S}_i$ are assumed to be spin-$\frac{1}{2}$ vector operators given by 
\begin{equation}\label{eq:spin_op}
\mathbf{S}_{i} =  [S_i^x, S_i^y, S_i^z], \quad S_i^v = I\otimes I \otimes \cdots \otimes I \otimes \underbrace{\frac{\hslash}{2} \sigma_v}_{i\text{th site}}\otimes I \otimes \cdots \otimes I  , \quad v = x,y,z
\end{equation}
for each spin at site $i$,
where $\sigma_x, \sigma_y, \sigma_z\in \mathrm{SU}(2) $ (the special unitary matrices of size 2) are the Pauli matrices. We adopt the standard computational spin-$\frac{1}{2}$ (qubit) basis $\ket{0}=\ket{\uparrow}$ and $\ket{1}=\ket{\downarrow}$, such that $\sigma_z\ket{i}=(-1)^{i}\ket{i}$ for $i=0, 1$. Throughout the paper, the following constant values $\mu_B=5.8\times 10^{-2}\ \text{meV/T}$, $\hslash=0.658\ \text{meV.ps}$, and $g=2$ are used.  

The Hamiltonian \eqref{eq:H1} represents a generalized form of the simplest two-spin (qubit) system discussed in \cite{liu24}. To extend the numerical analysis to a many-body system, the Hamiltonian is considered on a triangular lattice with nearest-neighbor interactions, as illustrated in Figure \ref{fig1}. The competition between the exchange interaction and the DMI can give rise to topological spin textures such as skyrmions, which are typically stabilized by a nonzero magnetic field \cite{Sotnikov2020}.

\begin{figure}[h]
\begin{center}
\includegraphics[width=70mm]{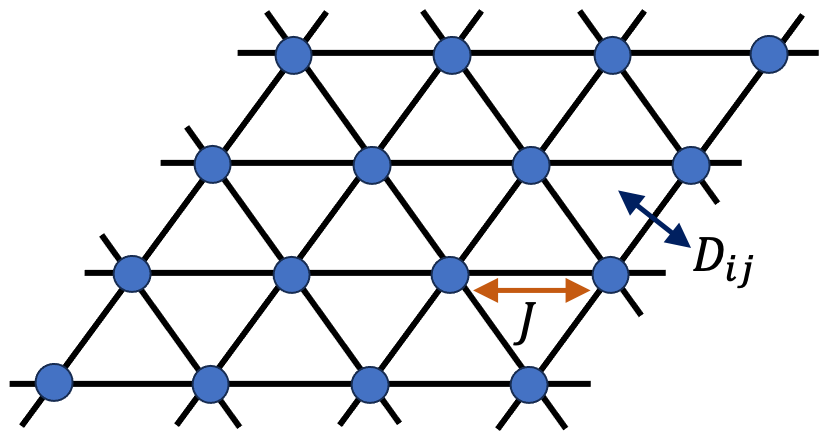}
\end{center}
\caption{Schematics of a triangular lattice with $J$ representing the Heisenberg exchange interaction, and ${D}_{ij}$ representing in-plane DMI vectors in directions perpendicular to bonds between neighboring $i$ and $j$ spin sites.}
\label{fig1}
\end{figure}

To emphasize the role of nonlinear effects in the quantum LLG equation, the dissipation parameter is set to be $\kappa = 0.5$ throughout the analysis. In the many-body case, without loss of generality, we have chosen to work with a 16-spin clusters with periodic boundary conditions. We also choose the following class of density matrices as the initial state $\rho(0)$ of the q-LLG dynamics:
\begin{equation*}
\rho(0) = \frac{p_0}{2^{n}} I + p_1 \ket{\text{AF}_{1}}\bra{\text{AF}_{1}} + p_2 \ket{\text{AF}_{2}}\bra{\text{AF}_{2}} + p_3 \ket{\text{GHZ}}\bra{\text{GHZ}} + p_4 \ket{W}\bra{W}
\end{equation*}
for a probability vector $p = (p_0, p_1, p_2, p_3, p_4)$, where the probabilities $p_i \in [0, 1]$ sum to unity. The antiferromagnetically ordered states are given by
\begin{equation*}
\ket{\text{AF}_i} = \bigotimes_{l=1}^{n} \ket{(l + i) \bmod 2}, \quad i = 1, 2,
\end{equation*}
while the generalized GHZ and $W$ states are defined as
\begin{equation*}
\ket{\text{GHZ}} = \frac{1}{\sqrt{2}}\left( \ket{0}^{\otimes n} + \ket{1}^{\otimes n} \right), \quad \quad \ket{W} = X \ket{0}^{\otimes n},
\end{equation*}
where 
$X = \frac{2}{\hslash\sqrt{n}} \sum_{i=1}^{n} S_i^x$, with $S_i^x$ defined in \eqref{eq:spin_op}.

Let us first examine the accuracy and convergence order of various standard RK methods using a pure (rank-1) state $\rho_0$ as an initial condition.  
Note that for the special case of a rank-1 initial density matrix $\rho_0$, there exists an exact solution to \eqref{eq:q-LLG}  given by \cite{liu24}
\begin{equation}\label{rho_exact}  
\rho(t) = \frac{\exp\left( -\frac{\ii}{\hslash}\tilde Ht\right)\rho_0 \exp\left(\frac{\ii}{\hslash}\tilde Ht\right)}{\mathrm{Tr}\left( \exp\left( -\frac{\ii}{\hslash}\tilde Ht\right)\rho_0 \exp\left(\frac{\ii}{\hslash}\tilde Ht\right)\right)},
\end{equation}
with $\tilde H = \left(\frac{1 - \ii\kappa}{1 + \kappa^2} \right)H$. This exact solution  
can be used for comparison with the corresponding numerical results.

Following \cite{liu24}, we begin with the simplest case of tow-spin ($n=2$) system
with Hamiltonian parameters $J = 1$ meV, $D_{21}=-D_{12}=|D|(0, 0, 1)$, $|D|=0.4$ meV, and $|B|=1$ T being the strength of magnetic field vector $B = |B|(1,0,0)$. Given the pure state initial condition, $\rho_0 =\ket{\text{AF}_{1}}\bra{\text{AF}_{1}}=\ket{01}\bra{01}$, Figure\ \ref{fig:error_cmp} presents the error plots for RK methods \eqref{rk:sstageform} and the conservative RK methods \eqref{rk:sstageform_modified}, for orders 1 through 4. The conservative RK family maintains full order accuracy across all timestep levels, whereas the errors in the classical RK3 and RK4 methods begin to deteriorate for small values of $h$. This test case involves a small quantum LLG problem with a $4 \times 4$ density matrix; for larger systems, such numerical instabilities with the non-conservative methods would likely to become even more pronounced. 

\begin{figure}[ht!]
    \centering
    \includegraphics[width=0.35\linewidth]{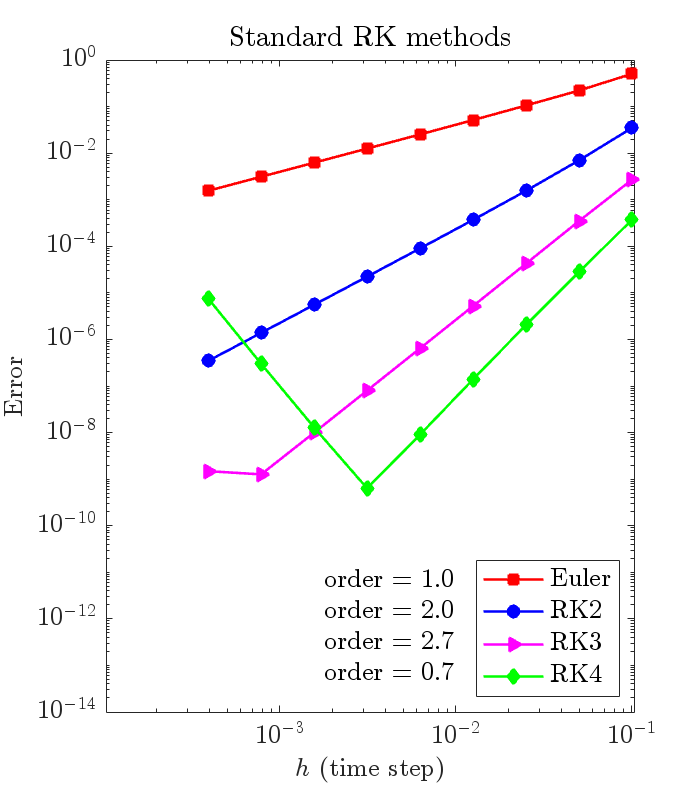}
    \includegraphics[width=0.35\linewidth]{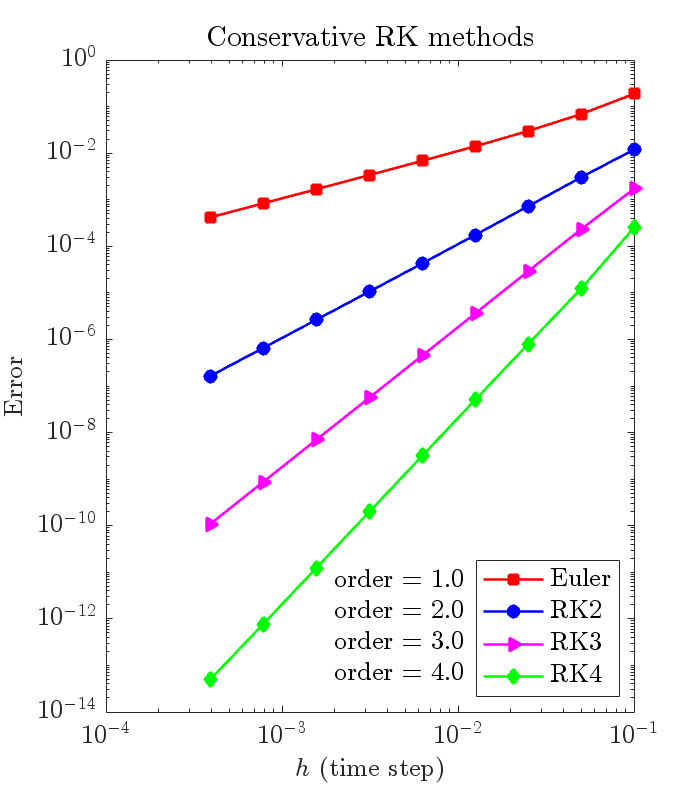}
    \caption{Error plots and convergence orders of standard RK methods (left) and the conservative RK methods (right) for the 2-spin system.}
    \label{fig:error_cmp}
\end{figure}

\begin{figure}[ht!]
    \centering
    \includegraphics[width=0.4\linewidth]{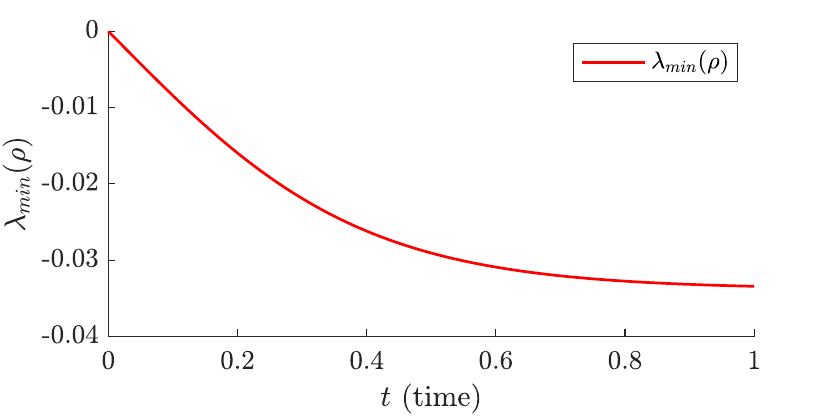}
    \includegraphics[width=0.4\linewidth]{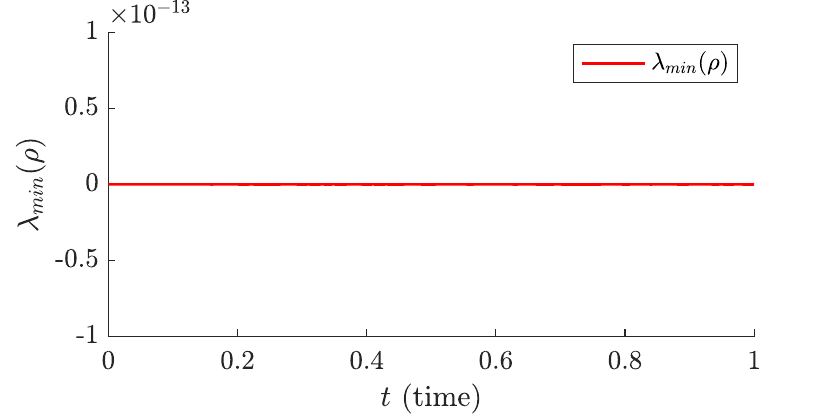} \\
    \includegraphics[width=0.4\linewidth]{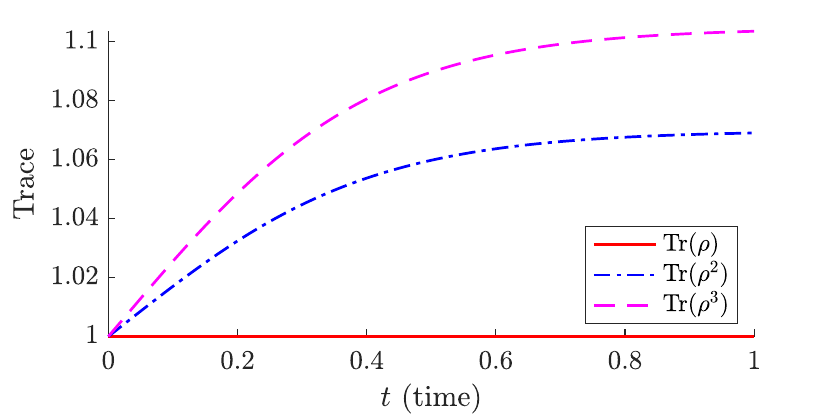}
    \includegraphics[width=0.4\linewidth]{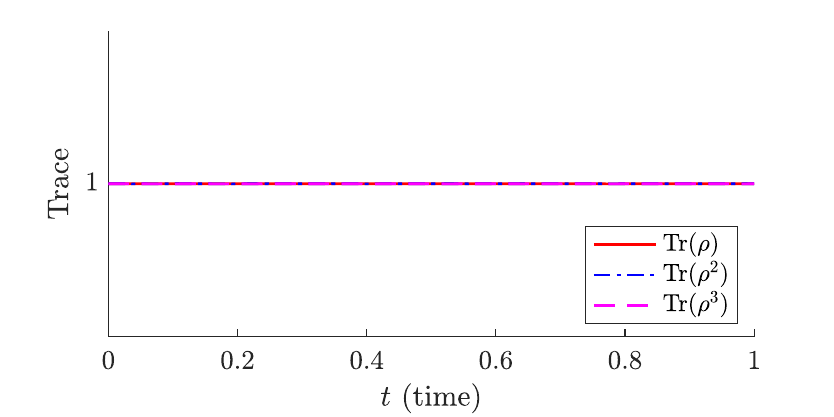}
    \caption{Minimum eigenvalues (upper panels) and trace values (lower panels) of numerical 2-spin density matrices in time interval $[0,1]$ using the standard Euler's method (left panels) and the conservative Euler's method (right panels).}
    \label{fig:eig_cmp}
\end{figure}

Figure \ref{fig:eig_cmp} shows the minimum eigenvalue and traces of various powers of the numerical solutions $\rho_k$, obtained from both method families. The standard RK methods produce solutions with negative eigenvalues, violating the semi-positive definiteness required for the density matrix. Additionally, they fail to preserve the traces of higher powers of $\rho_k$. However, these deficiencies are resolved by the conservative RK methods, which preserve both non-negativity and trace properties throughout the simulation.

We extend the results on convergence to a 16-spin cluster on a triangular lattice with periodic boundary conditions for the pure-state initial condition $\rho_0=\ket{\text{AF}_{1}}\bra{\text{AF}_{1}}$ in Figure \ref{fig:spin16_orders}. The same values and directions for Hamiltonian parameters as in the 2-spin case are used. Note that here $\rho_0$ is a rank-1 density matrix of dimension $2^{9}\times 2^{9}$ with entries
$\rho_0(i,j) = \delta_{i,171}\,\delta_{171,j}$, where \(\delta_{kl}\) is the Kronecker delta. The exact solution for this case is computed using formula \eqref{rho_exact}.
As we observe from Figure \ref{fig:spin16_orders}, the higher-order standard RK methods fail to attain their theoretical convergence rates, while the conservative methods achieve the expected orders. For brevity, we have omitted the corresponding figures showing that the standard methods also fail to preserve non-negativity and traces.

\begin{figure}[ht!]
    \centering
    \includegraphics[width=0.35\linewidth]{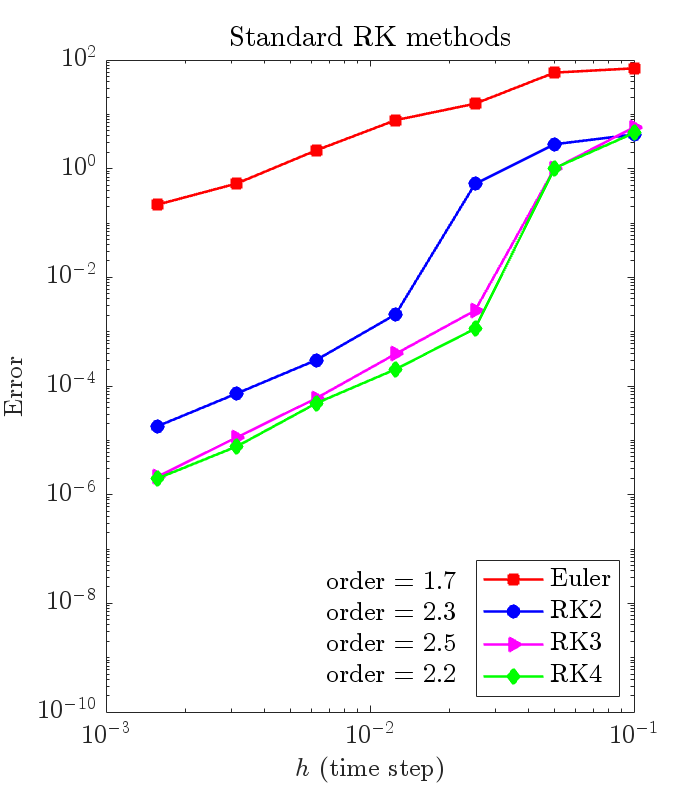}
    \includegraphics[width=0.35\linewidth]{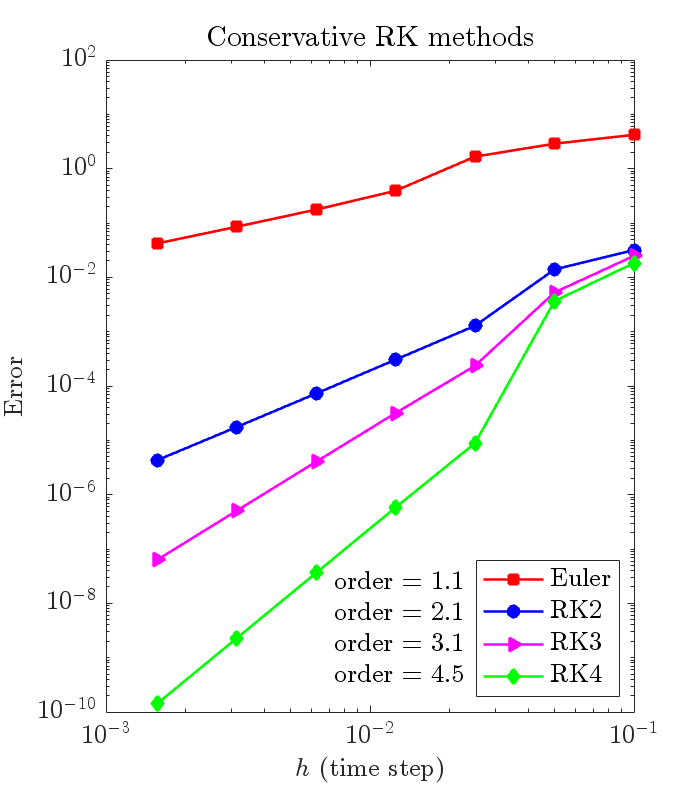} 

    \caption{Error plots and convergence order of standard RK methods (left) and the conservative RK methods (right) for the 16-spin system.} 
    \label{fig:spin16_orders}
\end{figure}

Although the results presented above correspond to specific parameter values and particular directions of the DMI and magnetic field, similar convergence plots have been observed for a variety of parameter settings and directions, which are not shown here. 

Having demonstrated the accuracy of the numerical procedure, we now apply it to investigate the dynamics of key physical properties of the 16-spin system (Figure \ref{fig1}), such as energy, magnetization, and entanglement, for both pure states and more general mixed state initial conditions where exact solutions are not available. 
For energy and magnetization, we examine the expectation values 
$$
\langle A\rangle=\Tr[A\rho(t)]
$$ 
where the observable $A$ is either the energy operator $H$ or the magnetization operators $M_v = \frac{1}{n} \sum_{i=1}^{n} S_i^v$, with $v = x, y, z$ and $S_i^v$ denoting the spin operator at site $i$ in the $v$ direction given in \eqref{eq:spin_op}. For entanglement, we consider the concurrence $C[\rho_{kl}(t)]$ (see the appendix for the definition) \cite{Vidal2002} of the reduced two-spin density matrix $\rho_{kl}(t)=\Tr_{\{\mathbf{S}_i\}_{i=1\&i\ne k,l}^{n}}\rho(t)$, where trace is taken over all degrees of freedom (basis states) of the spins not in $\{\mathbf{S}_k, \mathbf{S}_l\}$, leaving only the reduced state of the $k^{\text{th}}$ and $l^{\text{th}}$ spins \cite{Nielsen2010}. These measurements quantify the entanglement between two spins at the given sites $k$ and $l$. Naturally, different spin pairs yield different entanglement values. However, due to the translational symmetry of the system, the entanglement between any two adjacent spins is the same. Therefore, we restrict our analysis to a single pair of adjacent spins, namely, the spins at sites one and two. All results below are obtained with the conservative RK4 method with time step $h = 0.02$ for the 16-spin system. Unless otherwise specified, we use $|J| = 1~\text{meV}$, in-plane $D_{i,j}$ vectors perpendicular to the bonds between neighboring spins at sites $i$ and $j$ with uniform magnitude $|D_{i,j}| = 0.8~\text{meV}$, and a uniform magnetic field along the $z$-direction with strength $|B| = 1~\text{T}$.

Figure \ref{fig:spin16-ex1} shows the energy dissipation with the initial condition
$\rho_0$ being one of the following antiferromagnetically ordered states $\rho(0)=\ket{\text{AF}_{1}}\bra{\text{AF}_{1}}$, $\rho(0)=\ket{\text{AF}_{2}}\bra{\text{AF}_{2}}$, and $\rho(0)=\frac{1}{2}\ket{\text{AF}_{1}}\bra{\text{AF}_{1}}+\frac{1}{2}\ket{\text{AF}_{2}}\bra{\text{AF}_{2}}$. The plots clearly confirm the dissipative nature of the nonlinear term in the quantum LLG equation \eqref{eq:q-LLG}. Although results are shown here for only a few initial conditions, the energy dissipation under quantum LLG dynamics remains consistent regardless of the initial state.

\begin{figure}[ht!]
    \centering
    \includegraphics[width=0.45\linewidth]{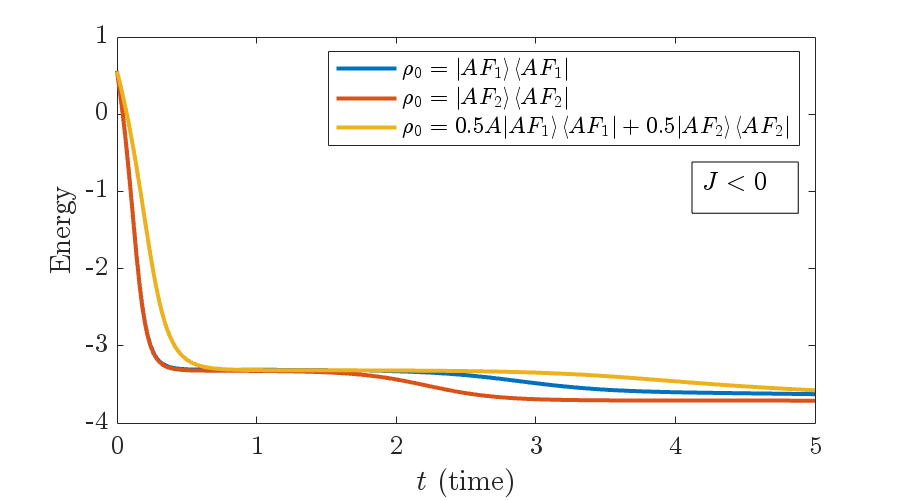}
    \includegraphics[width=0.45\linewidth]{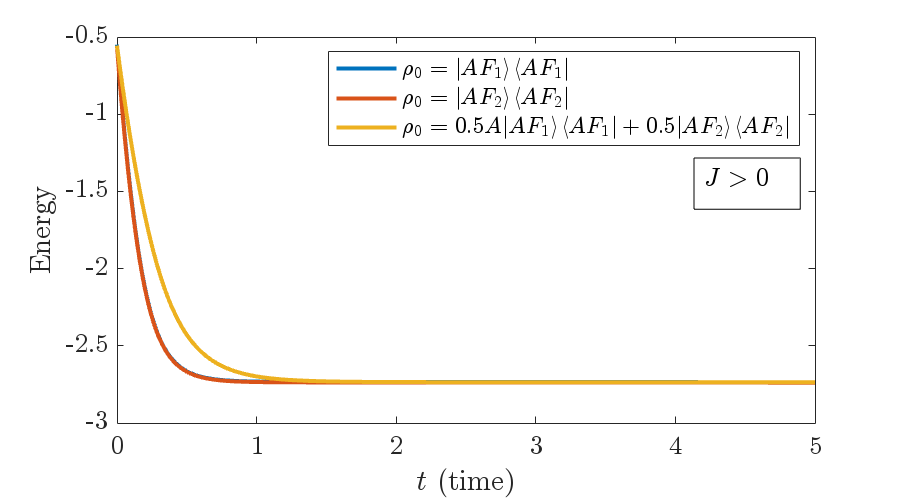}
    \caption{Quantum LLG dynamics of the energy expectation value, $\langle H \rangle = \Tr[H\rho(t)]$, where $H$ is the Hamiltonian of the system, are shown for different initial conditions. The dissipative nature of quantum LLG dynamics is independent of the initial state. Ferromagnetic (antiferromagnetic) Heisenberg exchange coupling is represented by $J<0$ ($J>0$). Energy and time are expressed in millielectronvolts (meV) and picoseconds (ps), respectively.
} 
    \label{fig:spin16-ex1}
\end{figure}

\begin{figure}[ht!]
    \centering
    \includegraphics[width=0.45\linewidth]{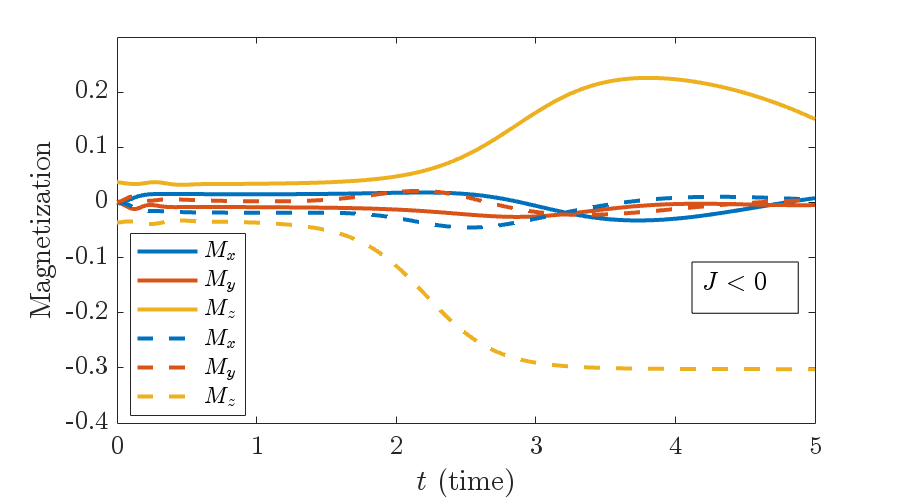}
    \includegraphics[width=0.45\linewidth]{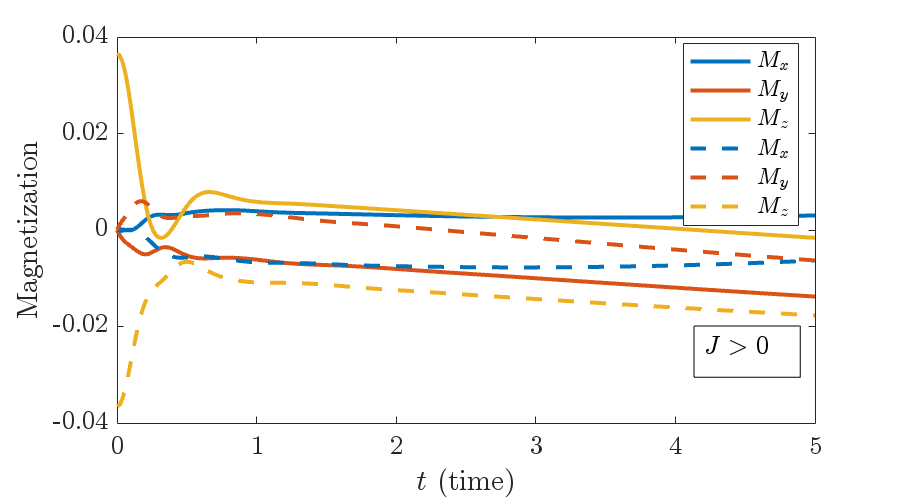} \\  
    \includegraphics[width=0.45\linewidth]{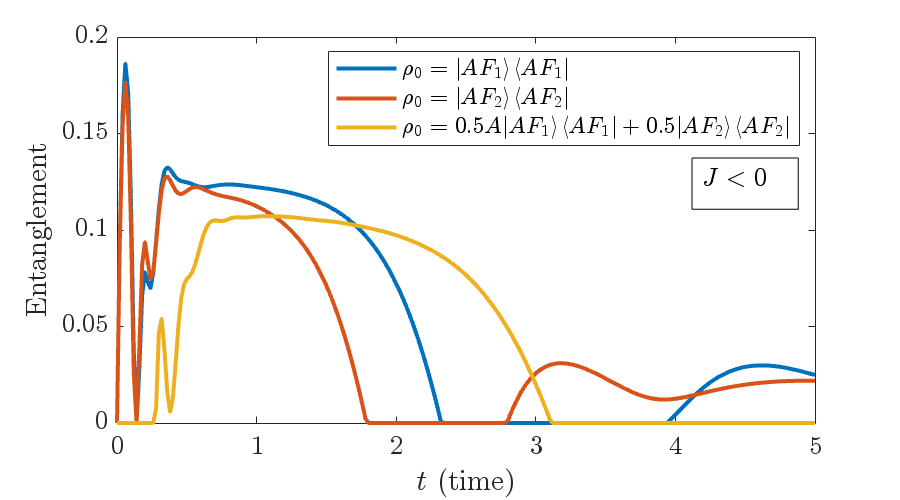}
    \includegraphics[width=0.45\linewidth]{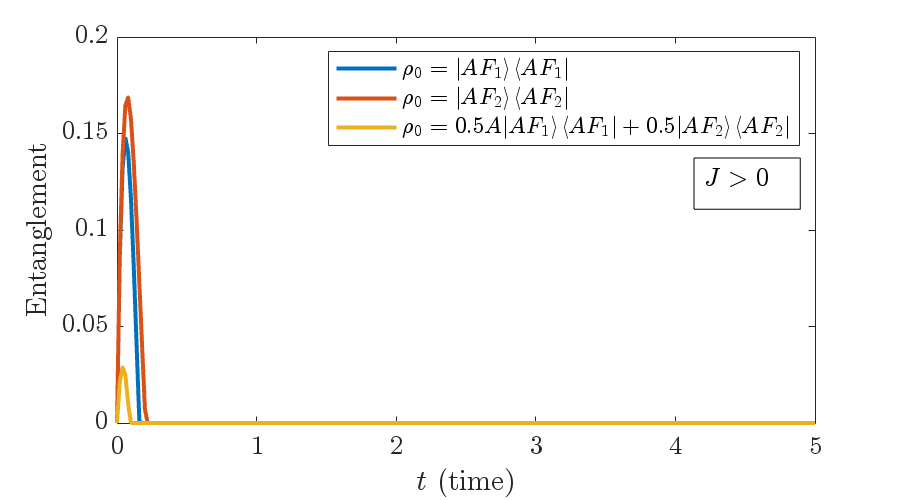} \\
    \caption{Dynamics of magnetization and spin–spin entanglement using the same initial conditions as in Figure \ref{fig:spin16-ex1}. Upper panels: average magnetization components, where solid and dashed curves correspond to the $\rho(0)=\ket{\text{AF}_1}\bra{\text{AF}_1}$ and $\rho(0)=\ket{\text{AF}_2}\bra{\text{AF}_2}$ initial conditions, respectively. Lower panels: concurrence quantifying bipartite entanglement between two nearest-neighbor spins. Ferromagnetic (antiferromagnetic) Heisenberg exchange coupling is represented by $J<0$ ($J>0$). Time is measured in picoseconds (ps).} 
    \label{fig:spin16-ex11}
\end{figure}

\begin{figure}[ht!]
    \centering
    \includegraphics[width=0.45\linewidth]{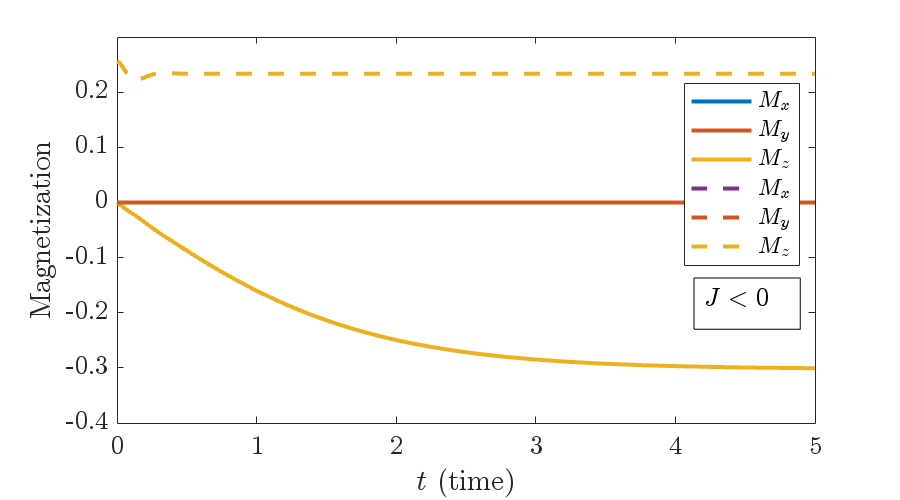}
    \includegraphics[width=0.45\linewidth]{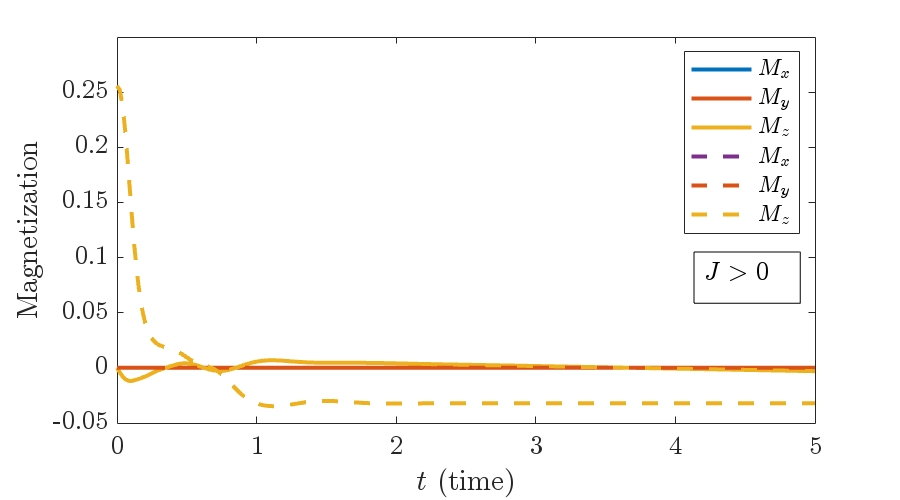} \\  
    \includegraphics[width=0.45\linewidth]{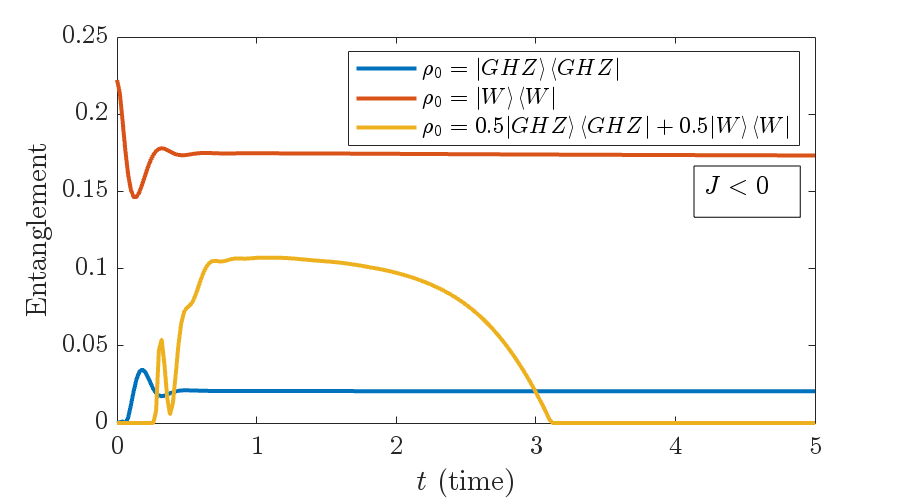}
    \includegraphics[width=0.45\linewidth]{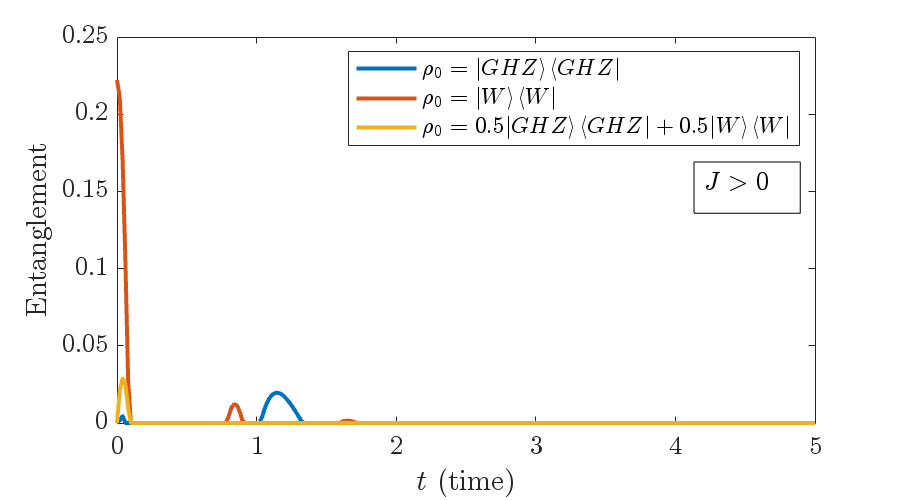} \\
    \caption{Dynamics of magnetization and spin–spin entanglement using $\rho(0)=\ket{\text{GHZ}}\bra{\text{GHZ}}$, $\rho(0)=\ket{W}\bra{W}$, and $\rho(0)=\frac{1}{2}\ket{\text{GHZ}}\bra{\text{GHZ}}+\frac{1}{2}\ket{W}\bra{W}$ as the initial state $\rho_0$. Upper panels: average magnetization components, where solid and dashed curves correspond to the GHZ and $W$ initial conditions, respectively. Lower panels: concurrence quantifying bipartite entanglement between two nearest-neighbor spins. Ferromagnetic (antiferromagnetic) Heisenberg exchange coupling is represented by $J<0$ ($J>0$). Time is measured in picoseconds (ps).} 
    \label{fig:spin16-ex2}
\end{figure}

We also analyze the dynamics of magnetization components and concurrence values for various initial conditions in Figures \ref{fig:spin16-ex11} and \ref{fig:spin16-ex2}. While Figure \ref{fig:spin16-ex11} employs the same initial conditions as Figure \ref{fig:spin16-ex1}, Figure \ref{fig:spin16-ex2} presents results for the GHZ and $W$ states, as well as for a mixed state defined by a convex combination of the two density matrices.

The developed numerical method enables systematic investigation of the effects of various mathematical and physical characteristics, such as intrinsic exchange couplings, external fields, system size, and dimensionality, on quantum spin dynamics governed by the quantum LLG equation. To illustrate some of these effects, we plot the dynamics of bipartite spin–spin entanglement (quantified via concurrence) for different values of spin–orbit DMI strength $|D|$ and magnetic field strength $|B|$ in Figure~\ref{fig:spin16-ex3}. The system is initially prepared in the mixed state
$\rho(0) = \tfrac{1}{2}\ket{\text{AF}_1}\bra{\text{AF}_1} + \tfrac{1}{2}\ket{\text{AF}_2}\bra{\text{AF}_2}$ with no bipartite entanglement between adjacent spins and the dynamics are considered in both the ferromagnetic ($J < 0$) and antiferromagnetic ($J > 0$) cases. Different entanglement dynamics are observed in the two cases. For instance, for a fixed magnetic field of moderate strength ($1\ T$), the entanglement in the antiferromagnetic case rises immediately as the evolution begins but dies out over a short time scale for any value of DMI strength. In contrast, in the ferromagnetic case, although the entanglement is initially zero, it exhibits a sudden rise after some time and persists for a considerable duration when the DMI strength is low. However, the lifetime of entanglement decreases as the DMI strength increases. Moreover, while a stronger magnetic field in the ferromagnetic case shortens the lifetime of non-zero entanglement, the entanglement dynamics in the antiferromagnetic case become more complex, exhibiting fluctuations as the magnetic field strength increases.
Such complex and unconventional entanglement dynamics under the quantum LLG equation suggest further and detailed studies on the physical nature and characteristics of these dynamics.

\begin{figure}[ht!]
    \centering
    \includegraphics[width=0.45\linewidth]{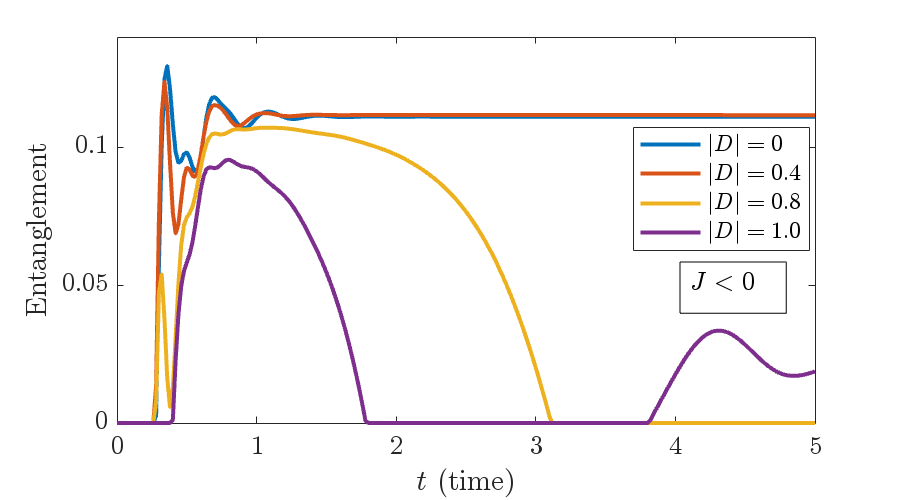}
    \includegraphics[width=0.45\linewidth]{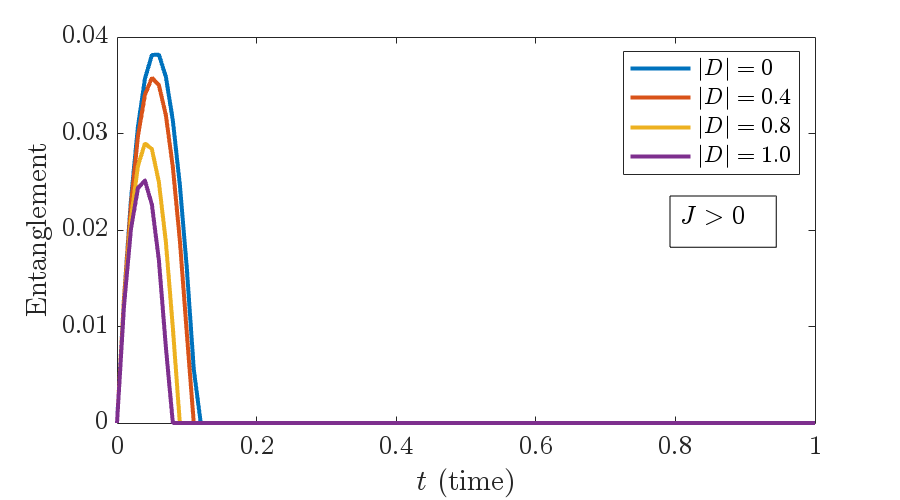}
    \includegraphics[width=0.45\linewidth]{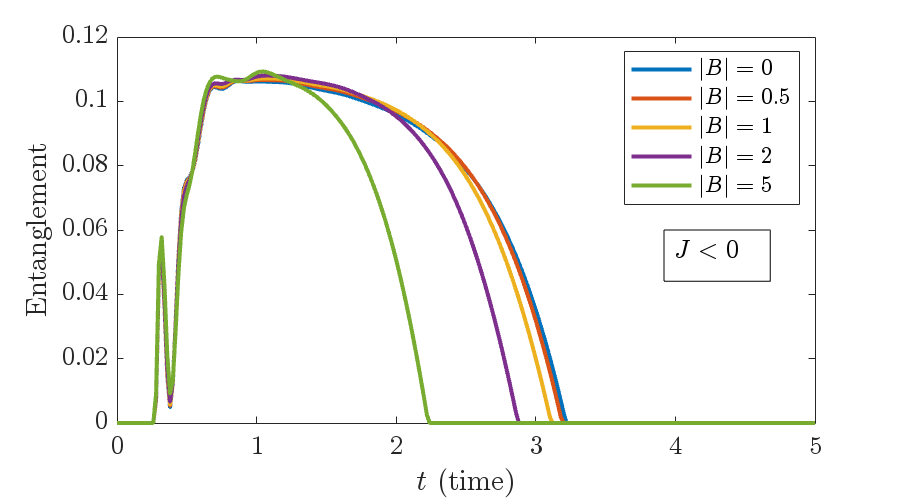}
    \includegraphics[width=0.45\linewidth]{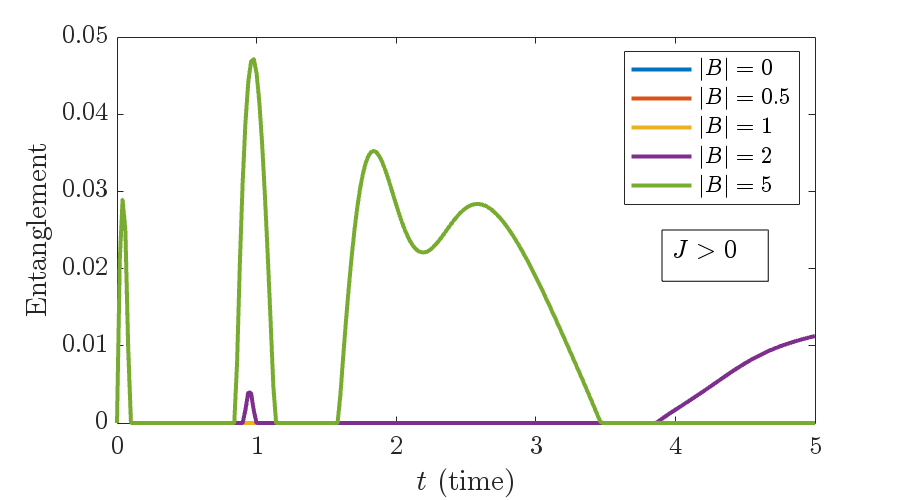}
    \caption{Bipartite spin–spin entanglement (concurrence) dynamics  for various spin–orbit DMI and magnetic field strengths, shown for both ferromagnetic ($J < 0$, left panels) and antiferromagnetic ($J > 0$, right panels) regimes. The upper panels show the effect of different DMI strengths on entanglement dynamics when the magnetic field is fixed at $|B| = 1$ T. Conversely, the lower panels illustrate how entanglement dynamics vary with different magnetic field strengths, while the DMI interaction is fixed at $|D| = 0.8$ meV. The initial condition is taken as $\rho(0) = \tfrac{1}{2}\ket{\text{AF}_1}\bra{\text{AF}_1} + \tfrac{1}{2}\ket{\text{AF}_2}\bra{\text{AF}_2}$. Time is measured in picoseconds (ps).}
    \label{fig:spin16-ex3}
\end{figure}

We also investigate the effect of system size on spin–spin entanglement dynamics by comparing the results of the 2-spin and 16-spin systems using the initial state $\rho(0)=\frac{1}{2}I/2^n+\frac{1}{2}\ket{W}\bra{W}$, as illustrated in Figure \ref{fig:spin16-ex4}. For the 2-spin case ($n=2$), the corresponding initial state is entangled, and the entanglement oscillates during the evolution, eventually retaining and stabilizing at its initial value. In contrast, for the 16-spin system ($n=9$ with periodic boundary condition), the initially zero bipartite entanglement remains unchanged throughout the time evolution under the quantum LLG dynamics. This observation is qualitatively the same in both the ferromagnetic and antiferromagnetic cases, indicating that system size plays a significant role in the emergence of entanglement dynamics.
The vanishing bipartite entanglement in the many-body setting suggests that entanglement may be distributed more globally across the system, underscoring the complexity of entanglement dynamics in larger quantum systems. This complexity is particularly interesting and motivates further theoretical and numerical investigation to fully understand the physical origin of entanglement behavior and its redistribution as the number of particles increases.

\begin{figure}[ht!]
    \centering
    \includegraphics[width=0.45\linewidth]{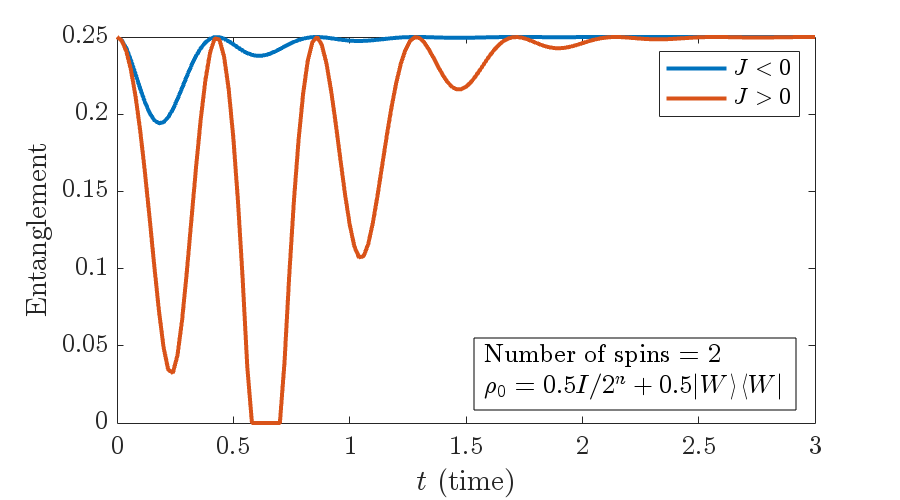}
    \includegraphics[width=0.45\linewidth]{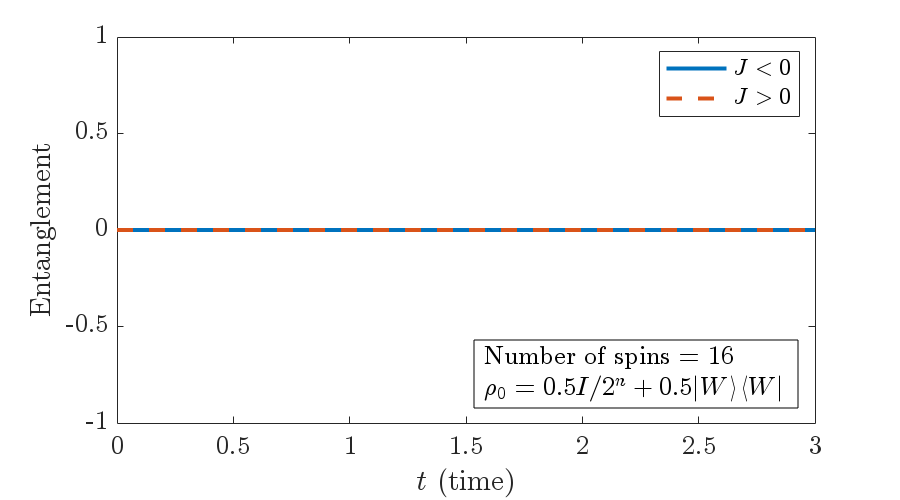}
    \caption{Bipartite spin–spin entanglement (concurrence) dynamics, plotted for both the 2-spin and 16-spin quantum systems. The magnetic field $B$ and the DMI vectors $D_{ij}$ are aligned along the $z$-axis, with fixed strengths $|D| = 0.8$ meV and $|B| = 1$ T. The sign of the Heisenberg exchange coupling $J$ distinguishes between ferromagnetic ($J < 0$) and antiferromagnetic ($J > 0$) spin systems. Time is measured in picoseconds (ps).}
    \label{fig:spin16-ex4}
\end{figure}

\section{Summary}
We have developed and demonstrated a numerical framework tailored to the recently proposed quantum generalization of the Landau–Lifshitz–Gilbert (LLG) equation which enables the simulation of quantum spin dynamics in many-body spin systems. Our methodology addresses the intrinsic challenges of quantum many-body dynamics by preserving essential physical and mathematical properties of the quantum state, including Hermiticity, trace preservation, and non-negativity. By applying the method to a class of quantum spin systems, we have observed unconventional and rich quantum phenomena, including the formation of long-lived entangled states and the emergence of spin correlations. These results highlight the capability of the quantum LLG equation to capture quantum behavior that lies beyond the reach of classical spin dynamics. 
This work contributes to a significant step toward scalable, structure-preserving simulations of quantum magnetism. It opens up new possibilities for investigating emergent quantum effects in complex magnetic systems and for exploring quantum technologies based on spin phenomena. 

\section*{Acknowledgments}
VAM gratefully acknowledges the insightful discussions with Prof. Olle Eriksson and his generous support through research funding.

\appendix
\section{Appendix}
Let $\rho \in \mathbb{C}^{4 \times 4}$ be the density matrix of a two-qubit (spin-$\frac{1}{2}$) system. The concurrence $C(\rho)$ is defined as \cite{Wootters1998}
\begin{equation}
C(\rho) = \max \left\{ 0, \lambda_1 - \lambda_2 - \lambda_3 - \lambda_4 \right\},
\end{equation}
where $\lambda_1, \lambda_2, \lambda_3, \lambda_4$ are the square roots of the eigenvalues (in decreasing order) of the matrix
\[
R = \rho \, \tilde{\rho}
\]
and $\tilde{\rho}$ is the spin-flipped state, given by
\[
\tilde{\rho} = (\sigma_y \otimes \sigma_y) \, \rho^* \, (\sigma_y \otimes \sigma_y). 
\]
Here $\rho^*$ is the complex conjugate of $\rho$ in the computational basis, and $\sigma_y$ is the Pauli-Y matrix.

\bibliographystyle{plainnat}
\bibliography{dmref}

\end{document}